\documentclass[12pt]{article}

\usepackage[numbers,sort&compress]{natbib}
\usepackage{setspace}
\usepackage{geometry}
\usepackage{graphicx}
\usepackage{booktabs}
\usepackage{amsmath}
\usepackage{amssymb}
\usepackage{hyperref}
\usepackage{enumitem}
\usepackage{authblk}

\usepackage{tikz}
\usetikzlibrary{arrows.meta, positioning, shapes, calc}

\usepackage{amsthm}
\newtheorem{definition}{Definition}  
\newtheorem{proposition}{Proposition}  
\newtheorem{lemma}{Lemma}  

\title{Norm-Governed Multi-Agent Decision-Making in Simulator-Coupled Environments:
The Reinsurance Constrained Multi-Agent Simulation Process (R–CMASP)}

\author[1]{Stella C. Dong\thanks{Corresponding author: \href{mailto:stella.dong@reinsuranceanalytics.io}{stella.dong@reinsuranceanalytics.io}}}

\affil[1]{Reinsurance Analytics, San Francisco, CA, USA}

\date{}
\begin{document}

\maketitle

\begin{abstract}
Reinsurance decision-making displays the core structural properties that motivate multi-agent models: distributed and asymmetric information, partial observability, heterogeneous epistemic responsibilities, environment dynamics mediated by external stochastic simulators, and hard prudential and regulatory constraints. Deterministic workflow automation cannot satisfy these requirements, as it lacks the epistemic flexibility, cooperative coordination mechanisms, and norm-sensitive behaviour necessary for institutional risk-transfer processes.

We propose the \emph{Reinsurance Constrained Multi-Agent Simulation Process} (\textsc{R--CMASP}), a formal model that extends stochastic games and Dec-POMDPs by incorporating three elements missing from prior MAS and agentic-LLM frameworks: (i) \emph{simulator-coupled transition dynamics} driven by catastrophe, capital, and portfolio engines; (ii) \emph{role-specialized epistemic agents} with structured observability, belief updates, and typed communication; and (iii) a \emph{normative feasibility layer} encoding solvency, regulatory, and organizational rules as admissibility constraints on joint actions. \textsc{R--CMASP} thereby provides a principled environment model for domains in which norms and simulators jointly determine feasible behaviour.

We instantiate the formalism using LLM-based agents equipped with tool access, structured prompts, and typed message protocols (state broadcasts, proposals, critiques, constraints). Experiments in a domain-calibrated synthetic environment show that governed multi-agent coordination yields more stable, consistent, and norm-adherent equilibria than both deterministic automation and monolithic LLM baselines—achieving lower pricing variance, higher capital efficiency, and improved clause-interpretation accuracy—while preserving auditability and human oversight. We further show that embedding prudential norms as admissibility constraints and structuring communication into typed acts measurably enhances equilibrium stability and normative adherence at the MAS level.

Beyond the reinsurance domain, our results support a broader methodological conclusion: regulated, simulator-driven decision environments are most naturally and effectively modelled as \emph{norm-governed, simulator-coupled multi-agent systems}. \textsc{R--CMASP} provides a general computational ontology for integrating epistemic reasoning, external simulators, constrained cooperation, and formal governance in high-stakes institutional settings.
\end{abstract}


\section{Introduction}
\label{sec:intro}

Reinsurance is a foundational mechanism of global risk transfer, enabling insurers to pool, cede, and diversify exposure to extreme natural catastrophe, cyber, liability, and other tail events \cite{cummins2012reinsurance, swissre2023sigma}. It supports solvency, stabilizes balance sheets, and mitigates accumulation risk, thereby contributing to systemic financial resilience. Yet the workflows required to interpret treaties, validate exposures, integrate stochastic hazard models, price risk, allocate capital, and steer portfolios are intrinsically complex. They combine (i) legally dense and often ambiguous contractual semantics, (ii) simulation-based hazard and loss models, and (iii) strict prudential regimes such as Solvency~II, IFRS~17, and NAIC~RBC that impose binding valuation, capital, and governance rules \cite{eiopa2020solvency2, ifrs17standard2020}.

Despite substantial tooling, reinsurance operations remain predominantly manual and epistemically fragmented. Underwriters, exposure analysts, catastrophe modelers, pricing actuaries, and capital specialists operate on heterogeneous information sets with partially conflicting objectives. Treaty interpretation requires resolving linguistic ambiguity \cite{mangini2019treatyanalysis}; pricing requires reconciling correlated hazard outputs \cite{grossi2005catmodels}; and portfolio steering must account for non-stationary perils and accumulation \cite{barnett2021nonstationary, mcgill2021aggregation}. These asymmetries routinely yield divergent assumptions, duplicated analyses, and elevated operational and model risk \cite{lloyds2019processreview}.

Classical workflow automation captures only the deterministic fragments of these processes. ETL pipelines, RPA systems, and ACORD-normalization frameworks \cite{rpa_insurance2018} function when transformations are unambiguous, but from a multi-agent systems (MAS) perspective they instantiate \emph{degenerate agents}: they have no belief states, cannot update interpretations, and cannot negotiate joint decisions under uncertainty. As argued in MAS and socio-technical research \cite{rahwan2019society_agents, wooldridge2009mas}, such systems fail in domains characterized by partial observability, distributed expertise, interdependent constraints, and binding norms—all intrinsic to reinsurance.

Recent developments in large language models (LLMs) \cite{brown2020gpt3, openai2024o1} and agentic reasoning frameworks—including reasoning--acting loops \cite{yao2022react}, metacognitive self-critique \cite{shinn2023reflexion}, and structured multi-agent protocols \cite{li2023autogen, wang2023surveyagents}—enable richer epistemic and organizational behaviour. These systems can parse heterogeneous documents, invoke external simulators, critique intermediate assumptions, and coordinate through structured message exchange. The epistemic and organizational structure of reinsurance aligns closely with these capabilities: workflows require semantic interpretation, simulator-grounded reasoning, constraint satisfaction, interdependent optimization, and explicit governance.

Crucially, reinsurance decision-making is not purely technical but profoundly \emph{normative}. Decisions must satisfy regulatory constraints, internal risk appetite, and auditability requirements. Following governed MAS principles \cite{rahwan2019society_agents}, the goal is not to replace human underwriters or actuaries, but to augment their epistemic and organizational capacity. Agents serve as analytic collaborators—surfacing inconsistencies, enforcing cross-role coherence, automating repetitive reasoning steps, and maintaining transparent audit trails—while final authority remains with human experts.

\paragraph{Contributions.}
This paper develops a theoretically grounded, institutionally aligned framework for governed multi-agent decision-making in reinsurance. Our contributions are:

\begin{itemize}
    \item \textbf{Conceptual.}
    We formalize the fundamental distinction between deterministic automation and autonomous, communication-enabled agents, showing that the epistemic, stochastic, and normative structure of reinsurance renders fixed pipelines structurally inadequate. From a MAS perspective, we contribute a new class of \emph{simulator-coupled, norm-governed decision processes} together with an \emph{operational equilibrium notion} tailored to institutional settings.

    \item \textbf{Architectural.}
    We introduce a multi-agent architecture mirroring the organizational decomposition of reinsurance practice, with role-specialized agents for treaty interpretation, exposure reasoning, hazard modeling, pricing, capital evaluation, portfolio steering, claims handling, and governance. Typed communication (state, proposal, critique, constraint) enables cooperative reasoning under uncertainty and normative constraints.

    \item \textbf{Formal and empirical.}
    We propose the \emph{Reinsurance Constrained Multi-Agent Simulation Process} (\textsc{R--CMASP}), a formal model that augments stochastic games and Dec-POMDPs with simulator-coupled transitions, normative feasibility constraints, and structured inter-agent communication. In a calibrated synthetic environment, governed multi-agent coordination yields substantial improvements in pricing stability, capital efficiency, interpretive accuracy, and governance robustness relative to deterministic automation and monolithic LLM baselines.
\end{itemize}

The paper’s main contributions are therefore: (i) a formal model (\textsc{R--CMASP}), (ii) an aligned multi-agent architecture, and (iii) an empirical instantiation in a simulator-grounded environment; taken together, these demonstrate how simulator-coupled, norm-governed MAS can be engineered and analysed in high-stakes institutional settings.

\paragraph{Paper Outline.}
Section~\ref{sec:related} reviews relevant work.  
Section~\ref{sec:automation-vs-agents} contrasts deterministic workflows with agentic computation.  
Section~\ref{sec:architecture} presents the multi-agent architecture.  
Section~\ref{sec:formal-model} introduces the \textsc{R--CMASP} formalism.  
Section~\ref{sec:experiments} reports empirical evaluations.  
Section~\ref{sec:discussion} analyzes theoretical implications.  
Section~\ref{sec:conclusion} concludes.

\section{Related Work}
\label{sec:related}

Reinsurance decision-making integrates semantic interpretation, stochastic simulation, multi-objective optimization, and institutionally governed coordination. Developing a principled multi-agent formalism for this domain touches five relevant literatures: (i) deterministic workflow automation in insurance, (ii) large language models (LLMs) for semantic reasoning, (iii) agentic LLM systems with tool use and structured communication, (iv) foundational multi-agent systems (MAS) research on coordination, norms, and organizational design, and (v) governance mechanisms for high-stakes socio-technical systems. We synthesize these strands using MAS-theoretic concepts—epistemic reasoning, normative constraints, and distributed coordination—to identify a structural gap motivating \textsc{R--CMASP}.  

Research in normative and organizational multi-agent systems has long examined 
obligations, permissions, institutional rules, and coordination under uncertainty. 
Our model builds on this lineage while integrating features—external simulators 
and prudential constraints—not addressed in prior MAS frameworks.

\subsection{Insurance Automation}
\label{sec:related-automation}

Operational automation in insurance is dominated by deterministic pipelines, RPA systems, and ACORD-standardized processing \cite{rpa_insurance2018}. In MAS terms, these systems instantiate \emph{degenerate agents}: they implement fixed transition rules, maintain no epistemic state, and cannot revise beliefs or form intentions in the BDI sense \cite{wooldridge2009mas}. They perform well only when data are schema-preserving and semantics unambiguous.

Reinsurance workflows violate these assumptions. They require (i) resolving linguistic and legal ambiguity in treaty text \cite{mangini2019treatyanalysis}, (ii) integrating heterogeneous stochastic hazard and capital models \cite{grossi2005catmodels}, and (iii) coordinating across organizational roles operating under asymmetric information. Deterministic automation fails under exactly these conditions, producing inconsistencies, duplicated analyses, and elevated model and operational risk \cite{lloyds2019processreview, cummins2012reinsurance}. This motivates a move from fixed workflows to systems capable of epistemic reasoning and cooperative coordination.

\subsection{Large Language Models}
\label{sec:related-llms}

LLMs provide powerful semantic grounding and approximate belief updating \cite{brown2020gpt3, openai2024o1}. They can extract treaty structure, reason over text, and perform multi-step chains of thought.

However, LLMs are \emph{not} agents in the MAS sense:  
they lack persistent epistemic state, explicit action models, shared world-state representations, commitment strategies, and formal coordination mechanisms. They cannot reason about other agents' knowledge or enforce institutional norms. Thus, while essential for treaty interpretation, LLMs alone do not provide the machinery required for structured multi-agent decision-making.

\subsection{Tool-Using and Agentic LLM Systems}
\label{sec:related-agentic}

Frameworks such as ReAct \cite{yao2022react}, Reflexion \cite{shinn2023reflexion}, and AutoGen \cite{li2023autogen} extend LLMs toward agent-like behaviour via tool invocation, self-critique, and structured dialogues. These systems approximate classical MAS ideas—distributed planning, critique cycles, and joint intention formation \cite{shoham1995mas, kraus1997mas}.

Yet they remain \emph{domain-general}. They do not:
\begin{itemize}
    \item embed domain-specific catastrophe, capital, or portfolio simulators directly in the transition model;
    \item represent multi-objective institutional rewards shaped by solvency or regulatory constraints;
    \item incorporate normative feasibility conditions or governance mechanisms as first-class components.
\end{itemize}

These gaps directly correspond to missing elements in classical normative 
and organizational MAS frameworks.

\subsection{Multi-Agent Systems in Finance}
\label{sec:related-mas}

The MAS literature offers deep foundations for coordination, distributed 
constraint satisfaction, and normative behaviour \cite{wooldridge2009mas, 
shoham1995mas}. Prior work has extensively examined  
(i) role-based and organizational MAS,  
(ii) normative systems with obligations, permissions, and sanctions, and  
(iii) negotiation and coordination under uncertainty.  
These strands provide the conceptual basis for modelling cooperative decision 
processes involving heterogeneous roles, partial observability, and 
institutionally governed constraints.

In finance, MAS approaches model markets, trading behaviour, and learning dynamics \cite{lebaron2006agentfinance, hommes2021agentfinc}, typically in competitive settings with strategic interplay. Reinsurance differs structurally:

\begin{itemize}
    \item interactions are \emph{cooperative}, not adversarial;
    \item agents represent \emph{organizational roles}, not self-interested actors;
    \item feasibility is defined by \emph{prudential and regulatory constraints}, not market equilibria;
    \item transition dynamics depend on \emph{external simulators} (catastrophe, capital, portfolio).
\end{itemize}

Existing MAS models provide coordination and negotiation toolkits but do not capture the simulator-coupled, norm-governed nature of institutional reinsurance.

\subsection{Human--AI Governance and Oversight}
\label{sec:related-governance}

Governance research emphasizes alignment, traceability, auditability, and override mechanisms in autonomous systems \cite{rahwan2019society_agents}. Within MAS, this corresponds to integrating \emph{norms}, \emph{institutional rules}, and \emph{supervisory agents} to maintain global consistency.

Reinsurance is intensively regulated: solvency constraints, reporting requirements, and internal risk policies shape the feasible action space. Modern agentic LLM systems do not encode these constraints, and classical normative MAS models do not incorporate external stochastic simulators or capital formulas as part of the environment dynamics. This gap motivates formalism in which governance and simulator-based feasibility jointly constrain behaviour.

\subsection{Comparison of Prior Work}
\label{sec:related-comparison}

Table~\ref{tab:related-comparison} compares prior strands along MAS-theoretic dimensions essential for institutional reinsurance: semantic grounding, tool-mediated reasoning, multi-step coordination, constraint negotiation, normative compliance, auditability, and human oversight.

\begin{table*}[htpb!]
\centering
\scriptsize
\caption{Comparison of prior research strands with respect to MAS-theoretic capabilities required for institutional reinsurance.}
\label{tab:related-comparison}
\begin{tabular}{lcccccc}
\toprule
\textbf{Capability} &
\textbf{Insurance} &
\textbf{LLMs} &
\textbf{Agentic} &
\textbf{MAS in} &
\textbf{Governance} &
\textbf{R--CMASP} \\
& \textbf{Automation} & & \textbf{LLMs} & \textbf{Finance} & \textbf{Research} & \textbf{(ours)} \\
\midrule
Semantic treaty interpretation & -- & \checkmark & \checkmark & -- & -- & \checkmark \\
Tool-mediated hazard/capital reasoning & -- & -- & \checkmark & -- & -- & \checkmark \\
Multi-step workflow coordination & -- & -- & \checkmark & \checkmark & -- & \checkmark \\
Negotiation across objectives & -- & -- & \checkmark & \checkmark & -- & \checkmark \\
Normative/regulatory constraints & -- & -- & -- & -- & \checkmark & \checkmark \\
Auditability and safety & -- & -- & -- & -- & \checkmark & \checkmark \\
Human-in-the-loop oversight & -- & -- & -- & -- & \checkmark & \checkmark \\
\bottomrule
\end{tabular}
\end{table*}

\subsection{Gap Analysis}
\label{sec:related-gap}

The synthesis reveals structural deficiencies across all prior approaches:

\begin{itemize}
    \item Deterministic automation lacks semantic grounding, epistemic reasoning, and deliberative coordination.
    \item LLMs supply semantics but lack persistent beliefs, action models, and multi-agent coordination.
    \item Agentic LLM frameworks support communication and tool use but not normative feasibility or domain-specific constraints.
    \item MAS in finance address competitive markets, not cooperative, regulation-bound institutional workflows.
    \item Governance research specifies principles but does not operationalize them within simulator-coupled MAS architectures.
\end{itemize}

Importantly, while normative MAS models reason about obligations, permissions, 
and sanctions, they typically do \emph{not} integrate external stochastic 
simulators or prudential capital rules as explicit components of the environment 
dynamics. \textsc{R--CMASP} incorporates both elements directly—treating 
simulation outputs and regulatory constraints as first-class drivers of 
state transitions—thereby bridging classical normative MAS with modern 
agentic systems that interact with quantitative models and institutional rules.

In combination, these gaps indicate the absence of a unified framework integrating semantic interpretation, simulator-coupled transitions, constrained multi-agent coordination, and explicit normative governance. This motivates the \textsc{R--CMASP} formalism introduced in Section~\ref{sec:architecture}.

\section{Automation vs.\ AI Agents}
\label{sec:automation-vs-agents}

Reinsurance decision-making integrates semantic interpretation, stochastic simulation, multi-objective coordination, and governance under binding regulatory norms. From a multi-agent systems (MAS) perspective, this environment exhibits partial observability, distributed expertise, and normative constraints that make static, unilateral decision rules inadequate. This section sharpens the distinction between \emph{deterministic automation} and \emph{autonomous agents}, drawing on epistemic logic, organizational MAS theory, and cooperative game-theoretic models.

\subsection{Deterministic Automation}
\label{sec:deterministic-automation}

Classical workflow automation instantiates a \emph{degenerate} form of agency. Such systems implement fixed, context-insensitive transition functions
\[
s_{t+1} = f(s_t),
\]
with no explicit epistemic state, deliberation, or communication. RPA, ETL pipelines, and deterministic workflow engines used in insurance operations exemplify this category \cite{rpa_insurance2018}.

\paragraph{Epistemic limitations.}
In epistemic logic, deterministic workflows lack knowledge ($K_i$) and belief ($B_i$) modalities. They cannot:
\begin{itemize}
    \item revise interpretations of ambiguous treaty text;
    \item incorporate stochastic evidence from hazard or capital simulators;
    \item represent uncertainty or distribution shift;
    \item reason about what other organizational roles know or assume.
\end{itemize}
Their epistemic state is effectively trivial, which makes them brittle when meaning, evidence, and assumptions evolve—precisely the setting of treaty interpretation, non-stationary hazard modeling, and capital evaluation.

\paragraph{Organizational and normative limitations.}
Normative MAS theory requires agents that can check obligations, permissions, and prohibitions and react to violations \cite{boella2009normative}. Deterministic systems with a fixed $f$ cannot:
\begin{itemize}
    \item verify solvency or capital constraints (e.g., SCR requirements) in a history-dependent way;
    \item enforce internal risk policies or escalation rules that depend on evolving evidence;
    \item reason about exceptions or conflicts between norms.
\end{itemize}
Encoding the space of regulatory and institutional contingencies into a single, static $f$ is infeasible; norm evaluation itself becomes an open-ended reasoning task.

\paragraph{Coordination and interdependence limitations.}
Reinsurance evaluation couples multiple interdependent tasks: treaty interpretation informs hazard modeling; hazard outputs inform pricing; pricing is constrained by capital and portfolio requirements. Deterministic automation lacks mechanisms for:
\begin{itemize}
    \item \emph{coordinating} across roles with asymmetric information;
    \item \emph{distributed constraint satisfaction} across interdependent objectives;
    \item \emph{proposal--critique cycles} characteristic of organizational decision-making \cite{kraus1997mas, jennings2000mas}.
\end{itemize}
Game-theoretically, such workflows have essentially \emph{zero cooperative capacity}: no joint intention formation, no search over Pareto trade-offs, and no adaptive response to normative or evidential feedback.

\paragraph{Deterministic workflows as a trivial R--CMASP.}
The above limitations can be expressed formally within our framework.

\begin{proposition}
\label{prop:degenerate-rcmasp}
Any deterministic workflow with fixed transition function $f : S \rightarrow S$ can be represented as an \textsc{R--CMASP} with a single agent, trivial epistemic state, and no communication. However, such a system cannot satisfy any norm that requires history-dependent or belief-dependent evaluation beyond $f$.
\end{proposition}

\begin{proof}[Sketch]
Construct $\mathcal{M} = \langle S, \{\mathcal{A}_1\}, T, \{\Omega_1\}, O, C, R, \mathcal{N}, \mathcal{U} \rangle$ with:
(i) a single agent $1$;  
(ii) $\mathcal{A}_1$ a singleton (the agent has no real choice);  
(iii) $T(s, a) = \delta_{f(s)}$ implementing the deterministic workflow;  
(iv) $\Omega_1 = S$ and $O(s,a) = \delta_s$ (full observation);  
(v) $C = \emptyset$ (no communication);  
(vi) $R$ arbitrary; and  
(vii) $\mathcal{N}, \mathcal{U}$ containing only norms that can be evaluated pointwise on $(s,a)$ without reference to history or beliefs.

This construction shows that any such workflow is an \textsc{R--CMASP} instance. Now consider any norm $N \in \mathcal{N}$ that requires (a) dependence on past states $(s_0,\dots,s_t)$ or (b) dependence on an agent belief $B_1$ that is not already encoded in $f$. Because $T$ and the policy are fixed, and no explicit history or belief state is maintained, $N$ cannot be evaluated or enforced other than by hard-coding its consequences into $f$. For open-ended, context-dependent norms (e.g., dynamic SCR constraints under changing scenarios), this is not possible in general. Hence such a system cannot, in any non-trivial sense, satisfy history- or belief-dependent norms.
\end{proof}

Proposition~\ref{prop:degenerate-rcmasp} formalizes the intuition that deterministic workflows are a strict special case of our model with trivial epistemic and normative structure. Institutional reinsurance, by contrast, requires non-trivial $\mathcal{N}$, $\mathcal{U}$, and communication.

\subsection{Autonomous AI Agents}
\label{sec:autonomous-ai-agents}

Autonomous agents differ fundamentally from deterministic systems by maintaining epistemic representations, deliberative policies, and communicative capabilities. Modern LLM-driven agents combine semantic grounding with multi-step planning, critique, and tool-mediated reasoning \cite{brown2020gpt3, yao2022react, shinn2023reflexion}. These capabilities align closely with MAS models of reasoning, organizational coordination, and norm compliance.

\paragraph{Epistemic and semantic agency.}
An autonomous agent $a_i$ maintains belief and knowledge states $(B_i, K_i)$ that support:
\begin{enumerate}
    \item semantic interpretation of treaty clauses and endorsements;
    \item belief revision in response to simulation feedback (hazard, capital, portfolio);
    \item higher-order epistemics (reasoning about others' knowledge), needed for coherent pricing and capital decisions;
    \item formation of mutual and common knowledge, central to cooperative agreement and joint intention.
\end{enumerate}
These features parallel dynamic epistemic logic and epistemic action models 
commonly used in multi-agent reasoning frameworks.

\paragraph{Organizational MAS perspective.}
Reinsurance workflows form an \emph{organizational MAS}: roles (interpretation, hazard modeling, pricing, capital, portfolio, claims) are defined by institutional structure rather than individual incentives. Agents must:
\begin{itemize}
    \item verify compliance with solvency and risk policies;
    \item reason about organizational rules, escalation paths, and exceptions;
    \item maintain justifications for auditability and accountability.
\end{itemize}
These behaviours arise naturally in architectures with typed communication, shared state, and explicit norms \cite{rahwan2019society_agents}.

\paragraph{Cooperative game-theoretic framing.}
Reinsurance decision-making corresponds to a \emph{cooperative, norm-constrained optimization} problem:
\begin{itemize}
    \item utility is shared (pricing accuracy, capital efficiency, solvency preservation);
    \item constraints are shared and hard (SCR limits, concentration bounds, governance rules);
    \item contributions are distributed across specialized roles;
    \item interaction proceeds via proposals, critiques, and constraint checks.
\end{itemize}
This structure fits distributed constraint optimization (DCOP) and joint-intention theory \cite{kraus1997mas, jennings2000mas} more naturally than competitive multi-agent games.

\paragraph{Capabilities of autonomous agents.}
Modern LLM-based agents support core components of cooperative MAS design:
\begin{itemize}
    \item \textbf{Semantic grounding} of treaty text and contractual logic;
    \item \textbf{Long-horizon deliberation} via reasoning--acting loops \cite{yao2022react};
    \item \textbf{Self-critique and metacognition} via Reflexion \cite{shinn2023reflexion};
    \item \textbf{Typed, structured communication} via AutoGen-style protocols \cite{li2023autogen};
    \item \textbf{Tool-mediated inference} through catastrophe, capital, and portfolio simulators;
    \item \textbf{Norm enforcement} through supervisory or governance agents \cite{rahwan2019society_agents}.
\end{itemize}

\paragraph{Why reinsurance strongly motivates agent-based formulations.}
The structural properties of reinsurance align with canonical MAS conditions:
\begin{itemize}
    \item \emph{Partial observability}: roles access different subsets of the global state;
    \item \emph{Distributed expertise}: interpretation, modeling, pricing, and capital evaluation are epistemically distinct;
    \item \emph{Hard normative constraints}: decisions must satisfy solvency, reserving, concentration, and auditability requirements;
    \item \emph{Interdependent objectives}: pricing interacts with capital and portfolio feasibility;
    \item \emph{Communication-dependent coordination}: consistent outcomes require proposal, critique, and consensus formation.
\end{itemize}

Taken together, these features \emph{strongly motivate} agent-based formulations rather than purely deterministic pipelines. Reinsurance decision-making is a prototypical example of a governed, cooperative MAS, which motivates the architecture described in Section~\ref{sec:architecture} and the \textsc{R--CMASP} formal model in Section~\ref{sec:formal-model}.

\section{Multi-Agent Architecture for Reinsurance}
\label{sec:architecture}

Reinsurance operations instantiate a prototypical \emph{organizational} and \emph{norm-governed} multi-agent system (MAS): decision-making is distributed across roles (underwriting, pricing, capital, portfolio, claims), information is heterogeneous and partially observable, and all actions are constrained by institutional and regulatory norms. In this section we describe a multi-agent architecture that operationalizes these properties and serves as a concrete instantiation of the \textsc{R--CMASP} formalism later developed in Section~\ref{sec:formal-model}.

We build on organizational MAS and role-based design \cite{wooldridge2009mas, dignum2004role}, normative systems and institutional reasoning \cite{boella2009normative}, distributed constraint optimization \cite{modi2005dcop}, and epistemic multi-agent reasoning \cite{fagin2003reasoning}. At a high level, the architecture provides: (i) a role structure aligned with reinsurance workflows, (ii) an institutional layer encoding norms and constraints, (iii) a typed communication protocol for cooperative reasoning, and (iv) LLM-based agent policies that implement semantic, epistemic, and normative behaviour.

Formally, we model the architecture as
\[
\mathcal{A} = \langle R,\, \mathcal{I},\, \mathcal{N},\, \mathcal{C},\, \Pi \rangle,
\]
where:
\begin{itemize}
    \item $R$ is a finite set of \emph{roles}, each corresponding to a functional responsibility (e.g., Treaty Interpretation, Pricing, Capital, Portfolio, Claims, Governance);
    \item $\mathcal{I}$ is an \emph{institutional model} capturing the organizational context: permissible workflows, escalation paths, and the mapping from workflows to roles;
    \item $\mathcal{N}$ is a set of \emph{norms} (obligations, permissions, prohibitions) induced by regulation, internal policies, and risk appetite;
    \item $\mathcal{C}$ is a \emph{communication protocol} specifying message types, conversation structures, and logging requirements;
    \item $\Pi = \{\pi_r : r \in R\}$ is a family of role-specific policies implemented by LLM-based agents with tool access and deliberative capabilities.
\end{itemize}

This tuple provides the organizational and normative scaffolding within which the \textsc{R--CMASP} state, transition, observation, and reward structures (Section~\ref{sec:formal-model}) are instantiated.

\paragraph{Architectural choices vs.\ formal requirements.}
Some aspects of $\mathcal{A}$ reflect domain-specific design choices, such as the concrete role set $R$ (e.g., separating Portfolio and Capital) and workflow mapping $\mathcal{I}$. Others are directly motivated by the \textsc{R--CMASP} formal model: the presence of a Governance Agent and Audit Trail Agent traces to the explicit norm sets $\mathcal{N}$ and feasibility constraints $\mathcal{U}$; the typed message protocol $\mathcal{C}$ mirrors the communication graph and speech-act types embedded in the formal model. Thus, the architecture is both \emph{domain-aligned} and \emph{formally grounded}.

\paragraph{Implementation note.}
In the empirical instantiation (Section~\ref{sec:experiments}), each role $r \in R$ is realized by an LLM-based agent with a prompt-specified policy and controlled tool access (hazard simulators, capital engines, retrieval systems). However, \textsc{R--CMASP} itself is agnostic to the underlying policy representation: $\Pi$ could equally comprise hand-crafted rules, learned policies (e.g., via multi-agent reinforcement learning), or hybrids.

\subsection{Organizational and Functional Layering}
\label{sec:layering}

Following role-based and organizational MAS frameworks \cite{dignum2004role, wooldridge2009mas}, we structure the architecture into four functional layers. Each layer groups roles that share epistemic focus and normative responsibilities:

\begin{enumerate}
    \item \textbf{Interpretation and Knowledge Layer.}  
    Agents transform unstructured artefacts (treaties, exposure files, historical documents) into structured epistemic content. Their primary function is semantic grounding and knowledge consolidation for the rest of the system.

    \item \textbf{Risk and Modeling Layer.}  
    Agents interface with stochastic simulators and quantitative models (catastrophe, capital, portfolio). They translate epistemic content into probabilistic beliefs and loss distributions, thereby shaping the stochastic environment of \textsc{R--CMASP}.

    \item \textbf{Decision and Optimization Layer.}  
    Agents coordinate to solve a constrained cooperative decision problem over pricing, capital, retrocession, and portfolio exposures. They operate as a distributed constraint optimization process \cite{modi2005dcop} under the norms $\mathcal{N}$.

    \item \textbf{Governance and Oversight Layer.}  
    Agents enforce institutional norms, verify cross-agent consistency, and decide when human oversight is required. They realize the normative and governance aspects of the institutional model $\mathcal{I}$ and norms $\mathcal{N}$ \cite{boella2009normative, rahwan2019society_agents}.
\end{enumerate}

Each role $r \in R$ is instantiated as an autonomous agent $a_r$ with:
\begin{itemize}
    \item an epistemic state $(K_r, B_r)$ capturing the agent's knowledge and beliefs;
    \item a normative profile (obligations, permissions, prohibitions) inherited from $\mathcal{N}$;
    \item a policy $\pi_r$ mapping local observations and messages into actions and outgoing messages.
\end{itemize}
The resulting architecture can be viewed as an organizational MAS in which agents jointly seek norm-compliant, utility-improving decisions under uncertainty.

\subsection{Workflow-Scoped Agent Activation}
\label{sec:workflow-activation}

Reinsurance processes are naturally segmented into workflows (e.g., pricing a new treaty, assessing a claim, optimizing retrocession), each with specific informational inputs, applicable norms, and optimization goals. To capture this, we introduce a \emph{workflow-scoped activation policy}
\[
\Phi : W \rightarrow 2^{R},
\]
mapping each workflow $w \in W$ to the subset of roles $R_w = \Phi(w)$ required for valid execution.

Activation depends on:
\begin{enumerate}
    \item \textbf{Epistemic pre-conditions}: which artefacts and simulations are available (treaty text, exposures, claims notices, etc.);
    \item \textbf{Normative requirements}: which regulatory and internal checks must be enforced (e.g., SCR impact for pricing);
    \item \textbf{Optimization scope}: which objectives are primary (e.g., capital efficiency versus portfolio diversification);
    \item \textbf{Tool dependencies}: which external simulators or engines must be invoked (hazard models, capital engines, portfolio aggregators).
\end{enumerate}

Table~\ref{tab:workflow-agents} summarizes typical role activations per workflow.

\begin{table}[htbp!]
\centering
\footnotesize
\caption{Workflow-scoped activation of roles in the multi-agent architecture.}
\label{tab:workflow-agents}
\begin{tabular}{p{4.5cm} p{10.5cm}}
\toprule
\textbf{Workflow} & \textbf{Activated Roles} \\
\midrule
Pricing & Treaty Interpretation, Exposure Understanding, Hazard Modeling, Pricing, Capital, Portfolio Steering, Governance \\
Claims Evaluation & Treaty Interpretation, Claims, Scenario/Footprint, Audit Trail, Governance \\
Retrocession Optimization & Scenario/Stress, Hazard Modeling, Capital, Retrocession Strategy, Portfolio Steering, Governance \\
Exposure Management & Exposure Understanding, Hazard Modeling, Scenario/Stress, Portfolio Steering, Capital, Governance \\
Regulatory Reporting & Regulatory Compliance, Audit Trail, Governance, Human Oversight \\
\bottomrule
\end{tabular}
\end{table}

From an organizational MAS viewpoint, each workflow $w$ induces a temporary sub-organization on $R_w$ with its own communication topology, local constraints, and decision responsibilities.

\subsection{Running Example: Pricing a Coastal Wind Treaty}
\label{sec:running-example}

To make the architecture concrete, we sketch a stylized pricing episode for a coastal wind excess-of-loss treaty. The environment provides: (i) a draft treaty wording with layered structure and exclusions, (ii) an exposure file with coastal property risks, and (iii) current portfolio and capital positions.

\paragraph{Activation.}
The workflow-scoped policy $\Phi$ activates roles
\begin{align*}
R_{\text{pricing}} =
\{\text{Treaty Interpretation}, \text{Exposure Understanding}, \text{Hazard Modeling}, \\ \text{Pricing}, \text{Capital}, \text{Portfolio Steering}, \text{Governance}\}.
\end{align*}

\paragraph{Epistemic grounding.}
The Treaty Interpretation Agent parses the wording to infer attachment points, limits, hours clauses, and wind-specific exclusions, and broadcasts a \textsf{State} message with the structured treaty representation. The Exposure Understanding Agent normalizes and validates the exposure file, producing a structured exposure state and issuing a corresponding \textsf{State} message.

\paragraph{Simulation and risk quantification.}
The Hazard Modeling Agent invokes coastal wind catastrophe models, generating annual loss distributions and event sets for the treaty. These populate $S_{\text{hazard}}$ and are communicated via \textsf{State} messages. The Capital Agent computes incremental SCR and solvency ratios under the proposed participation; the Portfolio Steering Agent computes changes in accumulation and concentration metrics in relevant coastal zones.

\paragraph{Proposal--critique--constraint cycle.}
The Pricing Agent uses the structured treaty, hazard outputs, and market priors to issue a \textsf{Proposal} for rate-on-line and terms. The Capital Agent responds with \textsf{Constraint} messages if the proposal breaches solvency or risk-appetite norms; the Portfolio Steering Agent may issue \textsf{Critique} messages if concentration limits are approached. The Pricing Agent revises its proposal accordingly. This exchange continues for a small number of rounds until no unresolved \textsf{Critique} remains and all active \textsf{Constraint} messages are satisfied.

\paragraph{Governance and escalation.}
The Governance Agent checks for internal inconsistencies (e.g., mismatches between treaty interpretation and pricing assumptions, or unexplained deviations from historical practice). If unresolved inconsistencies or norm violations remain, it triggers the Human Oversight Agent, which compiles a summary for underwriters and risk committees. Otherwise, the episode terminates in an operational equilibrium (Section~\ref{sec:formal-model}), and the Audit Trail Agent logs the full interaction.

This running example illustrates how the same formal machinery---roles, message types, norms, and simulator calls---is instantiated in a realistic pricing scenario, and anticipates the formal definitions of state, transition, observation, reward, and feasibility in \textsc{R--CMASP}. We revisit a variant of this coastal wind treaty scenario in Section~\ref{sec:case-study} as a qualitative case study of empirical behaviour.

\subsection{Typed Communication and Interaction Protocol}
\label{sec:interaction-protocol}

Cooperative reasoning is enabled by a structured communication protocol $\mathcal{C}$ over a set of typed messages
\[
\mathcal{M} = \{\textsf{State},\, \textsf{Proposal},\, \textsf{Critique},\, \textsf{Constraint}\}.
\]
This design builds on speech-act inspired agent communication and negotiation frameworks \cite{kraus1997mas, jennings2000mas} and is consistent with message-passing patterns in recent LLM-based multi-agent systems \cite{li2023autogen}.

Each message type has a distinct epistemic and normative role:
\begin{itemize}
    \item \textsf{State}: broadcasts partial views of the global state (e.g., treaty parses, hazard outputs, capital positions), contributing to shared or common knowledge $C_G$ \cite{fagin2003reasoning};
    \item \textsf{Proposal}: encodes candidate decisions (rates, capital allocations, retrocession structures), functioning as publicly observable intentions;
    \item \textsf{Critique}: challenges or refines proposals, enabling distributed belief revision and cooperative constraint resolution;
    \item \textsf{Constraint}: propagates norms and feasibility restrictions (capital floors, risk limits, regulatory requirements) derived from $\mathcal{N}$.
\end{itemize}

Interaction proceeds in discrete rounds:
\begin{align*}
(1)\; \text{receive messages} \;\rightarrow\;
(2)\; \text{update epistemic and normative state} \;\rightarrow\; \\
(3)\; \text{invoke tools / perform inference} \;\rightarrow\;
(4)\; \text{issue new messages or actions}.
\end{align*}
This process implements a distributed constraint optimization mechanism \cite{modi2005dcop} over the constrained decision space defined by \textsc{R--CMASP}, while maintaining auditability via the Audit Trail Agent (Section~\ref{sec:operations}).

\subsection{Interpretation and Knowledge Agents}
\label{sec:interpretation}

Agents in the Interpretation and Knowledge Layer transform heterogeneous artefacts into structured epistemic representations.

\paragraph{Treaty Interpretation Agent.}
Uses LLM-based semantic parsing \cite{brown2020gpt3} to extract layered structures, triggers, exclusions, hours clauses, and other contractual features. Formally, it realizes a mapping from raw text to a structured treaty state $S_{\text{treaty}}$, thereby populating part of the global state in \textsc{R--CMASP}.

\paragraph{Exposure Understanding Agent.}
Normalizes and validates exposure data (e.g., geocoded risks, insured values), assembling $S_{\text{exposure}}$ and resolving missing or inconsistent entries.

\paragraph{Knowledge Retrieval Agent.}
Implements retrieval-augmented reasoning \cite{wang2023surveyagents} over internal corpora (historical treaties, past decisions, regulatory circulars), enabling other agents to condition their policies $\pi_r$ on institutional memory.

\subsection{Risk and Modeling Agents}
\label{sec:risk-modeling}

Agents in the Risk and Modeling Layer couple epistemic content to stochastic simulators and quantitative models, thereby shaping the transition dynamics of \textsc{R--CMASP}.

\paragraph{Hazard Modeling Agent.}
Interfaces with catastrophe models and other hazard simulators \cite{grossi2005catmodels}, selecting perils, regions, and model configurations. It outputs loss distributions and event sets that populate $S_{\text{hazard}}$.

\paragraph{Scenario and Stress Testing Agent.}
Generates multi-peril, correlated scenarios incorporating non-stationarity and emerging risks \cite{barnett2021nonstationary}. These scenarios support counterfactual reasoning and stress testing of portfolio and capital decisions.

\paragraph{Model Risk Agent.}
Assesses model uncertainty, calibration drift, and conflicting outputs across simulators. It may generate \textsf{Constraint} messages when model risk exceeds thresholds, introducing normative obligations to re-evaluate or escalate.

\subsection{Decision and Optimization Agents}
\label{sec:decision-optimization}

Agents in this layer perform cooperative decision-making over a constrained multi-objective space, consistent with DCOP-style formulations \cite{modi2005dcop}. Given a joint decision variable $\mathbf{x}$ (e.g., treaty pricing, capacity allocation, retrocession structure), they implicitly aim to solve:
\[
\max_{\mathbf{x}} \;
U_{\text{pricing}}(\mathbf{x}) + U_{\text{capital}}(\mathbf{x}) + U_{\text{portfolio}}(\mathbf{x})
\quad \text{s.t.} \quad \mathbf{x} \in \mathcal{F},
\]
where $\mathcal{F}$ is the feasible set induced by contractual semantics, hazard outputs, and norms $\mathcal{N}$.

\paragraph{Pricing Agent.}
Combines treaty structure, exposure, hazard outputs, and market priors to propose rate-on-line and structure-specific pricing recommendations. It emits \textsf{Proposal} messages and revises them in response to \textsf{Critique} and \textsf{Constraint} signals.

\paragraph{Capital Agent.}
Computes Solvency Capital Requirement (SCR), RBC measures, diversification benefits, and solvency ratios \cite{eiopa2020solvency2}. When proposals violate capital or solvency norms, it issues \textsf{Constraint} messages, narrowing the feasible set $\mathcal{F}$.

\paragraph{Portfolio Steering Agent.}
Evaluates portfolio-level accumulation, peril correlation, and concentration metrics \cite{mcgill2021aggregation}. It negotiates with the Pricing and Capital Agents to ensure that individual treaty decisions are consistent with portfolio-wide risk appetite.

\paragraph{Retrocession Strategy Agent.}
Evaluates quota-share, excess-of-loss, and aggregate retrocession structures, considering tail-risk mitigation, capital relief, and availability constraints. Conceptually, it operates as a cooperative partner in the joint optimization problem above.

\subsection{Operations, Compliance, and Claims Agents}
\label{sec:operations}

Operational and compliance agents ensure that decisions are correctly executed and institutionally valid.

\paragraph{Claims and Recoveries Agent.}
Interprets claims notifications, event footprints, and treaty triggers to compute recoveries at treaty and retrocession levels. It updates $S_{\text{claims}}$ and informs subsequent capital and portfolio updates.

\paragraph{Regulatory Compliance Agent.}
Checks that decisions adhere to external regulations (Solvency~II, IFRS~17) and internal policies encoded in $\mathcal{N}$. Violations trigger \textsf{Constraint} messages and may force re-negotiation.

\paragraph{Audit Trail Agent.}
Maintains a tamper-evident log of messages, decisions, and constraint applications. This provides an explicit record for ex-post audit, consistent with socio-technical governance requirements \cite{rahwan2019society_agents}.

\subsection{Governance and Human Oversight}
\label{sec:governance}

\paragraph{Governance Agent.}
Acts as a supervisory agent that monitors global consistency and norm satisfaction. It:
\begin{itemize}
    \item cross-checks outputs across agents for epistemic and numerical consistency;
    \item flags hallucinations and contradictions in LLM-generated content;
    \item enforces that the joint policy profile $\Pi$ remains within the normative constraints $\mathcal{N}$.
\end{itemize}
Formally, it seeks to maintain $\Pi \models \mathcal{N}$ and to trigger corrective processes whenever inconsistencies become common knowledge.

\paragraph{Human Oversight Agent.}
Implements a meta-level policy that determines when human intervention is required. It aggregates explanations, counterfactual analyses, and key metrics into summaries suitable for underwriters, actuaries, and risk committees. High-stakes or ambiguous cases are thus escalated to human decision-makers, preserving accountability and expert control.


Overall, this architecture realizes reinsurance decision-making as a governed organizational MAS, providing a concrete operational semantics for the \textsc{R--CMASP} framework. The next section formalizes the underlying state, transition, observation, reward, and constraint structures that define \textsc{R--CMASP} as a domain-grounded cooperative multi-agent decision process.

\section{Formal Model}
\label{sec:formal-model}

Reinsurance decision-making comprises a tightly coupled collection of processes---contract interpretation, exposure analysis, hazard modeling, pricing, capital evaluation, portfolio steering, claims assessment, and regulatory governance. Each component has been studied separately (e.g., catastrophe modeling \cite{grossi2005catmodels, barnett2021nonstationary}, solvency regulation \cite{eiopa2020solvency2}, portfolio aggregation \cite{mcgill2021aggregation}), but there is no unified formalism that treats reinsurance as a \emph{regulated cooperative multi-agent decision problem}.

Classical multi-agent system (MAS) models \cite{shoham1995mas, wooldridge2009mas, kraus1997mas}, stochastic games, and Dec-POMDPs provide rich foundations for joint decision-making under uncertainty, yet they typically abstract away  
(i) \emph{external simulators} (catastrophe, capital, and portfolio models) that define environment dynamics;  
(ii) \emph{contractual semantics} of legal artefacts such as treaties and endorsements; and  
(iii) \emph{prudential and regulatory constraints} that delimit the admissible joint action space.  

Conversely, recent LLM-based agent frameworks \cite{yao2022react, shinn2023reflexion, li2023autogen, wang2023surveyagents} provide mechanisms for tool use, multi-step deliberation, and message passing, but lack a principled treatment of solvency rules, normative constraints, and governance.

To bridge this gap, we introduce the \emph{Reinsurance Constrained Multi-Agent Simulation Process} (\textsc{R--CMASP}). Intuitively, an \textsc{R--CMASP} is a constrained, simulator-augmented multi-agent decision process in which  
(i) agents reason over structured representations of treaties, exposures, hazards, and capital;  
(ii) state transitions are mediated by external quantitative engines; and  
(iii) admissible policies are restricted by a normative layer encoding solvency, regulatory, and governance requirements.

\subsection{R--CMASP: Reinsurance Constrained Multi-Agent Simulation Process}
\label{sec:rcmasp}

Let $I = \{1,\dots,n\}$ index the agents.

\begin{definition}[Reinsurance Constrained Multi-Agent Simulation Process]
A Reinsurance Constrained Multi-Agent Simulation Process is a tuple
\[
\mathcal{M}
=
\bigl\langle
S,\,
\{ \mathcal{A}_i \}_{i \in I},\,
T,\,
\{\Omega_i\}_{i \in I},\,
O,\,
C,\,
R,\,
\mathcal{N},\,
\mathcal{U}
\bigr\rangle,
\]
where:
\begin{itemize}
    \item $S$ is the global state space;
    \item $\mathcal{A}_i$ is the action space of agent $i$ (with joint action space $\mathcal{A} = \prod_{i \in I} \mathcal{A}_i$);
    \item $T : S \times \mathcal{A} \rightarrow \Delta(S)$ is the transition kernel, implemented via external simulators and analytical engines;
    \item $\Omega_i$ is the observation space of agent $i$;
    \item $O : S \times \mathcal{A} \rightarrow \prod_{i \in I} \Delta(\Omega_i)$ is the joint observation kernel;
    \item $C \subseteq I \times I$ is a directed communication graph specifying which agents may send messages to which others;
    \item $R : S \times \mathcal{A} \rightarrow \mathbb{R}^k$ is a vector-valued reward capturing institutional objectives (e.g., capital efficiency, portfolio risk, cross-agent consistency, governance);
    \item $\mathcal{N}$ is a set of \emph{normative rules} (obligations, permissions, prohibitions) regulating joint behaviour;
    \item $\mathcal{U}$ is a set of \emph{exogenous regulatory and solvency constraints} (e.g., Solvency~II, IFRS~17) that any feasible policy profile must satisfy.
\end{itemize}
\end{definition}

\paragraph{Positioning.}
From a MAS perspective, \textsc{R--CMASP} can be viewed as a partially observable stochastic game or Dec-POMDP augmented with (i) a norm-induced feasibility component $\mathcal{F}$ derived from $\mathcal{N}$ and $\mathcal{U}$, and (ii) an explicit simulator coupling in $T$ that encodes domain-specific catastrophe, capital, and portfolio engines. When $\mathcal{N}$ and $\mathcal{U}$ are empty and $T$ abstracts away simulator calls, \textsc{R--CMASP} reduces to a standard partially observable stochastic game. In this sense, \textsc{R--CMASP} generalises existing formalisms by embedding normative feasibility and simulator-driven dynamics as first-class components of the environment.

Conceptually, an \textsc{R--CMASP} is thus a constrained, simulator-augmented, partially observable stochastic game equipped with an explicit normative layer \cite{wooldridge2009mas, boella2009normative}. The kernels $T$ and $O$ capture the stochastic, simulation-driven nature of reinsurance, while $\mathcal{N}$ and $\mathcal{U}$ encode institutional and regulatory structure typically absent from standard MAS and Dec-POMDP models.

\subsection{State Space}
\label{sec:state-space}

The global state $S$ decomposes into semantically meaningful components aligned with reinsurance workflows:
\[
S =
\bigl(
S_{\text{treaty}},
S_{\text{exposure}},
S_{\text{hazard}},
S_{\text{capital}},
S_{\text{portfolio}},
S_{\text{claims}},
S_{\text{regulatory}}
\bigr).
\]

Each component corresponds to a distinct epistemic dimension:
\begin{itemize}
    \item $S_{\text{treaty}}$: structured representations of treaties---layer structures, attachment points and limits, reinstatements, exclusions, hours clauses, event definitions, and ambiguous provisions---derived from textual wordings and endorsements \cite{mangini2019treatyanalysis}.
    \item $S_{\text{exposure}}$: geocoded exposures, insured values, line-of-business indicators, policy terms, and aggregate exposure statistics.
    \item $S_{\text{hazard}}$: simulation outputs from hazard and catastrophe models, including event catalogues, loss distributions, peril correlations, and scenario stress paths \cite{grossi2005catmodels, barnett2021nonstationary}.
    \item $S_{\text{capital}}$: Solvency Capital Requirement (SCR), Risk-Based Capital (RBC), diversification credits, solvency ratios, and internal capital metrics \cite{eiopa2020solvency2}.
    \item $S_{\text{portfolio}}$: accumulation measures, correlation matrices, diversification indices, and capacity utilisation across treaties and perils \cite{mcgill2021aggregation}.
    \item $S_{\text{claims}}$: reported and incurred losses, claim development factors, trigger states, and realised or potential recoveries at treaty and retrocession levels.
    \item $S_{\text{regulatory}}$: parameters and status variables capturing the current regulatory and governance environment (Solvency~II parameters, IFRS~17 discounting and contract boundaries, risk appetite statements, escalation thresholds).
\end{itemize}

Agents observe (possibly partial and noisy) projections of $S$ through $O$ and their role-specific observation spaces $\Omega_i$. For example, the Treaty Interpretation Agent is primarily informed by $S_{\text{treaty}}$, while the Capital Agent observes $(S_{\text{hazard}}, S_{\text{capital}}, S_{\text{portfolio}}, S_{\text{regulatory}})$.

\subsection{Agents, Local Policies, and Belief States}
\label{sec:agents-policies}

Each agent $i \in I$ maintains an internal belief state $b_i$ over $S$, updated via local observations, received messages, and tool outputs. Let $H_i$ denote the space of local histories for agent $i$, consisting of its past actions, observations, and incoming messages. A local policy is a mapping
\[
\pi_i : H_i \rightarrow \Delta(\mathcal{A}_i),
\]
which may be stochastic and history-dependent. In our implementation, $H_i$ is realised as the prompt history (including tool calls) of an LLM-based agent, and $\pi_i$ corresponds to a prompt-specified decision rule; the \textsc{R--CMASP} formalism also accommodates learned policies (e.g., multi-agent reinforcement learning).

The agent set reflects the role structure of Section~\ref{sec:architecture}:
\begin{itemize}
    \item \emph{Epistemic roles}: Treaty Interpretation, Exposure Understanding, Knowledge Retrieval;
    \item \emph{Simulation roles}: Hazard Modeling, Scenario/Stress Testing, Model Risk;
    \item \emph{Decision roles}: Pricing, Capital, Portfolio Steering, Retrocession Strategy;
    \item \emph{Operational and institutional roles}: Claims, Regulatory Compliance, Audit Trail;
    \item \emph{Governance roles}: Governance and Human Oversight \cite{rahwan2019society_agents}.
\end{itemize}

A joint policy profile $\boldsymbol{\pi} = (\pi_i)_{i \in I}$, together with $T$ and $O$, induces a probability distribution over trajectories in $S$ that must respect both the feasibility constraints induced by $\mathcal{N}$ and $\mathcal{U}$ and the communication structure $C$.

\subsection{Communication Graph and Message Semantics}
\label{sec:communication-graph}

Communication is represented by a directed graph
\[
C \subseteq I \times I,
\]
where $(i,j) \in C$ indicates that agent $i$ may send messages to agent $j$. Let $\mathcal{M}$ denote the set of message contents, with each message annotated by a type
\[
m.\text{type} \in \{\textsf{State},\, \textsf{Proposal},\, \textsf{Critique},\, \textsf{Constraint}\},
\]
as introduced in Section~\ref{sec:architecture}.

Messages have two primary formal roles:
\begin{enumerate}
    \item \textbf{Epistemic updates.}  
    \textsf{State} and \textsf{Proposal} messages act as public announcements \cite{fagin2003reasoning}, refining other agents' beliefs $b_j$ about the global state $S$ and about peers' intended actions.

    \item \textbf{Normative signalling.}  
    \textsf{Constraint} messages encode obligations, permissions, or prohibitions derived from $\mathcal{N}$ and $\mathcal{U}$ (e.g., capital floors, exposure limits), while \textsf{Critique} messages indicate non-acceptance or requested revision, guiding the system toward norm-compatible joint actions.
\end{enumerate}

In the implementation considered here, communication proceeds in synchronous rounds, but the formalism is agnostic to timing and admits asynchronous variants.

\subsection{Transition Kernel and Simulator Coupling}
\label{sec:transition-operator}

The transition kernel
\[
T : S \times \mathcal{A} \rightarrow \Delta(S)
\]
captures both deterministic and stochastic updates, including those mediated by external simulators. A single transition may consist of:
\begin{itemize}
    \item evaluating catastrophe or hazard models to generate loss distributions and event sets for selected perils and regions \cite{grossi2005catmodels, barnett2021nonstationary};
    \item running capital engines (e.g., Solvency~II standard formula or internal models) to update $S_{\text{capital}}$ \cite{eiopa2020solvency2};
    \item recalculating portfolio measures in $S_{\text{portfolio}}$ to reflect new binding decisions or retrocession structures \cite{mcgill2021aggregation};
    \item updating $S_{\text{claims}}$ with new reported or settled losses and associated recoveries;
    \item adjusting $S_{\text{regulatory}}$ when thresholds are breached or regulatory statuses change.
\end{itemize}

Thus, the environment dynamics in an \textsc{R--CMASP} are defined not only by abstract state-transition rules, but also by domain-calibrated quantitative engines that reflect external models and regulatory formulas.

\subsection{Reward Structure, Norms, and Feasibility}
\label{sec:reward-objectives}

The vector-valued reward $R$ assigns to each state--action pair $(s,\mathbf{a})$ a tuple
\[
R(s,\mathbf{a}) =
\bigl(
R^{\text{cap}}(s,\mathbf{a}),
R^{\text{port}}(s,\mathbf{a}),
R^{\text{cons}}(s,\mathbf{a}),
R^{\text{gov}}(s,\mathbf{a})
\bigr),
\]
where, for example:
\begin{itemize}
    \item $R^{\text{cap}}$ measures \emph{capital efficiency} (e.g., expected profit per unit of SCR);
    \item $R^{\text{port}}$ penalises \emph{portfolio risk}, accumulation, and concentration;
    \item $R^{\text{cons}}$ penalises \emph{cross-agent inconsistencies} (e.g., conflicting treaty interpretations or incompatible pricing and capital assessments);
    \item $R^{\text{gov}}$ penalises \emph{norm violations}, including breaches of prudential, regulatory, or governance rules \cite{rahwan2019society_agents}.
\end{itemize}

For analysis and evaluation, it is convenient to consider a scalarised global reward
\[
\tilde{R}(s,\mathbf{a})
=
\alpha R^{\text{cap}}(s,\mathbf{a})
-
\beta R^{\text{port}}(s,\mathbf{a})
-
\gamma R^{\text{cons}}(s,\mathbf{a})
-
\delta R^{\text{gov}}(s,\mathbf{a}),
\]
for nonnegative weights $\alpha, \beta, \gamma, \delta$.

Crucially, not all joint actions are admissible. The normative system $\mathcal{N}$ and exogenous constraints $\mathcal{U}$ induce a feasibility set
\[
\mathcal{F} \subseteq S \times \mathcal{A},
\]
such that any feasible policy profile $\boldsymbol{\pi}$ must satisfy
\[
\Pr_{\boldsymbol{\pi},T}
\bigl[
(s_t,\mathbf{a}_t) \in \mathcal{F}
\;\;\text{for all}\;\; t
\bigr]
= 1.
\]
Intuitively, norms and regulations act as hard constraints on behaviour, while $R$ (or $\tilde{R}$) shapes trade-offs within the feasible region.

\paragraph{Norm-induced policy restriction.}
It is useful to characterise how norms constrain behaviour relative to an unconstrained policy profile.

\begin{lemma}[Norm-induced policy restriction]
\label{lem:norm-restriction}
Let $\boldsymbol{\pi}$ be any joint policy profile on $\mathcal{A}$, and let $\mathcal{F} \subseteq S \times \mathcal{A}$ be the feasibility set induced by $(\mathcal{N},\mathcal{U})$. Assume that for every infeasible pair $(s,\mathbf{a}) \notin \mathcal{F}$ there exists at least one \emph{closest} feasible joint action $\mathbf{a}' \in \mathcal{A}$ according to some fixed tie-breaking rule. Then there exists a norm-compliant profile $\boldsymbol{\pi}'$ such that:
\begin{enumerate}
    \item for all $(s,\mathbf{a}) \in \mathcal{F}$, $\boldsymbol{\pi}'(\mathbf{a} \mid h) = \boldsymbol{\pi}(\mathbf{a} \mid h)$ whenever the history $h$ leads to $s$; and
    \item under $(\boldsymbol{\pi}',T)$, all visited state--action pairs lie in $\mathcal{F}$ almost surely.
\end{enumerate}
\end{lemma}

\begin{proof}[Proof sketch]
Construct $\boldsymbol{\pi}'$ by modifying $\boldsymbol{\pi}$ as follows. For any history $h$ leading to state $s$, if $\boldsymbol{\pi}$ assigns positive probability to an infeasible joint action $\mathbf{a}$ with $(s,\mathbf{a}) \notin \mathcal{F}$, reassign that probability to the corresponding closest feasible $\mathbf{a}'$ (according to the fixed tie-breaking rule). On feasible histories and actions, leave $\boldsymbol{\pi}$ unchanged. By construction, (i) holds by definition, and (ii) follows because all probability mass on infeasible actions is systematically redirected to feasible ones at each step, so trajectories under $(\boldsymbol{\pi}',T)$ remain in $\mathcal{F}$ with probability one.
\end{proof}

Lemma~\ref{lem:norm-restriction} formalises the intuitive view that norms and prudential rules restrict an otherwise unconstrained policy space to a norm-compliant subspace without changing behaviour on already-feasible decisions.

\subsection{Operational Equilibrium and Convergence}
\label{sec:equilibrium}

Because reinsurance decisions are episodic (e.g., pricing a treaty, assessing a claim), we adopt an \emph{operational} equilibrium notion tailored to constrained cooperative decision-making. Consider a workflow episode with a fixed joint policy profile $\boldsymbol{\pi}$ and the induced trajectory of states, actions, and messages. We say that the outcome is in \emph{\textsc{R--CMASP} equilibrium} if:

\begin{itemize}
    \item \textbf{Feasibility:}  
    all encountered $(s_t, \mathbf{a}_t)$ lie in $\mathcal{F}$; in particular, capital, portfolio, and regulatory constraints (as encoded in $S_{\text{capital}}$ and $S_{\text{regulatory}}$) are satisfied.

    \item \textbf{Consistency:}  
    there are no unresolved \textsf{Critique} messages; all active \textsf{Proposal} messages form a mutually compatible set (no conflicting rates, capital views, or portfolio decisions).

    \item \textbf{Constrained stability:}  
    given the communication graph $C$ and current beliefs, no single agent can unilaterally adjust its policy to achieve a strictly higher scalar reward $\tilde{R}$ without causing a violation of feasibility. Any further improvement in $\tilde{R}$ would require coordinated deviation by multiple agents.
\end{itemize}

This notion parallels negotiation-based equilibria in cooperative MAS \cite{kraus1997mas}, but explicitly incorporates the feasibility set $\mathcal{F}$ induced by norms and regulation, as well as simulator-driven dynamics.

\begin{proposition}[Existence of operational equilibrium under finite state and action spaces]
\label{prop:equilibrium-existence}
Assume that $S$ and $\mathcal{A}$ are finite, that norms $(\mathcal{N},\mathcal{U})$ induce a non-empty feasibility set $\mathcal{F}$, and that for any state $s$ and joint action $\mathbf{a}$ with $(s,\mathbf{a}) \notin \mathcal{F}$, the Governance and Compliance agents eventually emit \textsf{Constraint} messages that force revision to some feasible $\mathbf{a}'$. Suppose further that the proposal--critique--constraint process terminates in a finite number of rounds for any fixed $\boldsymbol{\pi}$. Then for any norm-compliant policy profile $\boldsymbol{\pi}$ (e.g., as in Lemma~\ref{lem:norm-restriction}), every episode admits at least one \textsc{R--CMASP} equilibrium outcome.
\end{proposition}

\begin{proof}[Proof sketch]
Under the stated assumptions, the proposal--critique--constraint dynamics over a finite state--action space induce a finite-state terminating negotiation process: repeated application of \textsf{Critique} and \textsf{Constraint} cannot continue indefinitely. Termination occurs only when (i) no further \textsf{Critique} or \textsf{Constraint} is emitted and (ii) the realised joint action lies in $\mathcal{F}$ by construction. At such a terminal point, feasibility and consistency are satisfied by definition. Constrained stability holds because any unilateral deviation that would increase $\tilde{R}$ while remaining feasible would, by assumption, have been reachable via further proposals and would thus contradict termination. Hence, the terminal outcome satisfies all conditions of a \textsc{R--CMASP} equilibrium.
\end{proof}

Proposition~\ref{prop:equilibrium-existence} does not claim uniqueness or global optimality; it formalises that, under mild termination assumptions, the governance and critique machinery suffices to induce well-defined operational equilibria for a given policy profile.

In this work, rather than proving stronger convergence guarantees, we empirically assess (in Section~\ref{sec:experiments}) (i) feasibility, (ii) consistency, and (iii) coordination cost (measured by the number of inter-agent message rounds) under fixed, prompt-specified policies.

\subsection{Graphical Overview}
\label{sec:formal-figure}

Figure~\ref{fig:formal-model} summarises the core components of the \textsc{R--CMASP}: a structured global state $S$, role-specialised agents communicating over $C$, and a global reward $\tilde{R}$ shaping cooperative behaviour under normative and regulatory constraints.

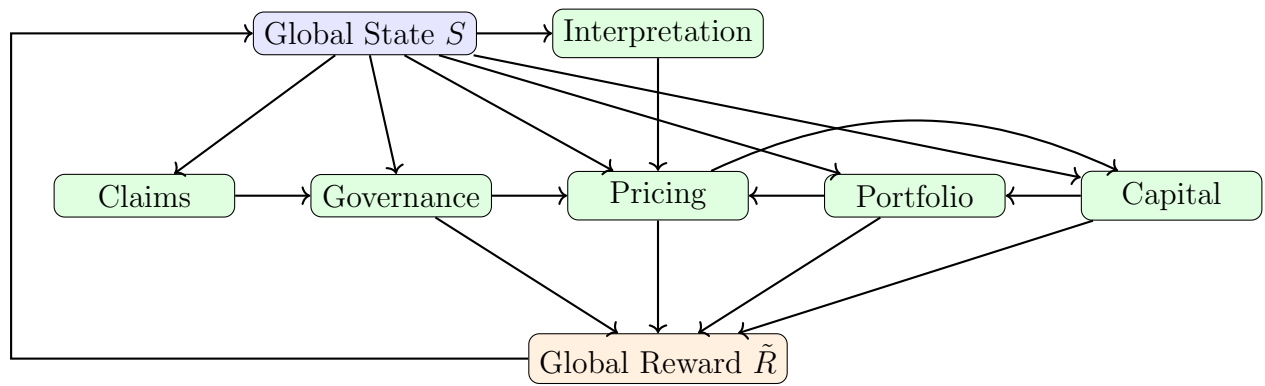
\begin{figure}[htpb]
\centering
\begin{tikzpicture}[
    node distance=1.5cm and 1cm,
    state/.style={rectangle, draw, rounded corners, fill=blue!10, minimum width=2.3cm},
    agent/.style={rectangle, draw, fill=green!12, rounded corners, minimum width=2.4cm},
    reward/.style={rectangle, draw, fill=orange!12, rounded corners, minimum width=3cm},
    arrow/.style={->, thick}
]

\node[agent] (a2) {Pricing};
\node[agent,   left = of a2] (a6) {Governance};
\node[agent, above  = of a2] (a1) {Interpretation};
\node[agent,  right = of a2] (a3) {Portfolio};
\node[agent,  left = of a6] (a5) {Claims};
\node[agent,  right =  of a3] (a4) {Capital};


\node[reward, below= of a2] (r) {Global Reward $\tilde{R}$};
\node[state, left = of a1] (s) {Global State $S$};

\draw[arrow] (s) -- (a1);
\draw[arrow] (s) -- (a2);
\draw[arrow] (s) -- (a3);
\draw[arrow] (s) -- (a4);
\draw[arrow] (s) -- (a5);
\draw[arrow] (s) -- (a6);
\draw[arrow] (a1) -- (a2);
\draw[arrow] (a3) -- (a2);
\draw[arrow] (a6) -- (a2);
\draw[arrow] (a2) to[bend left=25] (a4);
\draw[arrow] (a4) -- (a3);
\draw[arrow] (a5) -- (a6);

\draw[arrow] (a2) -- (r);
\draw[arrow] (a3) -- (r);
\draw[arrow] (a4) -- (r);
\draw[arrow] (a6) -- (r);
%
\draw[arrow] (r) -| (-8.6, -2) |- (s);

\end{tikzpicture}
\caption{Schematic view of the \textsc{R--CMASP}: global state, role-specialised agents, communication structure, and global reward shaping under normative and regulatory constraints.}
\label{fig:formal-model}
\end{figure}

\paragraph{Novelty.}
The \textsc{R--CMASP} formalism is, to our knowledge, the first to:
\begin{itemize}
    \item model reinsurance as a constrained, simulator-coupled multi-agent decision process with explicit state, observation, communication, reward, and normative components;
    \item integrate LLM-based semantic reasoning with external catastrophe, capital, and portfolio models within a single MAS framework;
    \item embed regulatory and governance norms as hard feasibility constraints on multi-agent trajectories;
    \item support role-specialised agents that coordinate via typed messages (\textsf{State}, \textsf{Proposal}, \textsf{Critique}, \textsf{Constraint}) over a communication graph $C$;
    \item provide an operational equilibrium notion that combines feasibility, consistency, and constrained stability for institutional risk-transfer decisions.
\end{itemize}

\section{Experiments}
\label{sec:experiments}

This section instantiates the proposed architecture within a controlled 
Reinsurance Constrained Multi-Agent Simulation Process (\textsc{R--CMASP}) and 
evaluates different joint policy profiles. Because real-world treaty and 
exposure data are proprietary and subject to strict confidentiality, we follow 
standard practice in catastrophe and portfolio risk modeling 
\cite{grossi2005catmodels, barnett2021nonstationary} and construct a synthetic 
but structurally realistic environment. This allows us to  
(i) control multi-peril dependence,  
(ii) impose explicit capital and portfolio constraints, and  
(iii) define exact ground truth for treaty semantics and feasibility.

Each experimental condition corresponds to a joint policy profile
\[
\boldsymbol{\pi} = (\pi_1,\dots,\pi_n),
\]
ranging from deterministic workflows to governed multi-agent coordination. 
All reported metrics are empirical functionals of the trajectory distribution 
induced by $(\mathcal{M}, \boldsymbol{\pi})$.

\subsection{Experimental Setup}
\label{sec:experimental-setup}

\paragraph{Synthetic portfolio generation.}
We generate a portfolio of 500 reinsurance treaties, sampling realistic 
attachment points, limits, reinstatements, line-of-business mixes, hours 
clauses, and exclusion combinations 
\cite{cummins2012reinsurance, swissre2023sigma}. The generator emits:
\begin{itemize}
    \item a textual representation simulating treaty wordings and schedules, and
    \item a structured representation (layering, triggers, exclusions, limits)
          that populates $S_{\text{treaty}}$.
\end{itemize}
Because both views are produced jointly, clause interpretation accuracy can be 
measured exactly by comparing agent outputs to this structured ground truth.

Exposure data are created at location level with geospatial clustering, insured 
values, and line-of-business labels, yielding realistic spatial and business-line 
structure in $S_{\text{exposure}}$.

\paragraph{Hazard and loss modeling.}
The hazard component $S_{\text{hazard}}$ consists of wind, flood, and wildfire 
modules with event catalogs, intensity footprints, vulnerability curves, and 
non-stationary perturbations \cite{grossi2005catmodels, barnett2021nonstationary}. 
We impose a prescribed cross-peril correlation structure to obtain multi-peril 
annual loss distributions and event sequences, which are observed (partially) by 
Pricing, Capital, and Portfolio agents.

\paragraph{Capital and regulatory constraints.}
Capital calculations approximate the Solvency~II standard formula 
\cite{eiopa2020solvency2}. For each portfolio configuration, we:
\begin{itemize}
    \item compute peril- and line-of-business-specific SCR components;
    \item aggregate components via prescribed correlation matrices to obtain
          diversified SCR;
    \item compute solvency ratios and compare them against minimum thresholds;
    \item enforce portfolio-level concentration and capacity constraints
          (e.g., maximum accumulation by zone, maximum tail VaR), following
          \cite{mcgill2021aggregation}.
\end{itemize}
These define a feasibility set $\mathcal{F} \subseteq S \times \mathcal{A}$ 
that is enforced by Capital, Portfolio, and Governance agents as part of their 
policies.

\paragraph{Agent configurations and baselines.}
We evaluate four configurations on the \emph{same} environment $\mathcal{M}$:
\begin{itemize}
    \item \textbf{Rule-based automation:}  
    a deterministic pricing--capital pipeline that applies fixed rating formulas 
    and capital loadings, with no role decomposition or inter-agent communication.

    \item \textbf{Single-agent LLM:}  
    a tool-using LLM that receives treaty text and exposure summaries, invokes 
    catastrophe and capital tools, and outputs pricing and accept/decline 
    decisions. Internally, the policy follows a ReAct-style reasoning--acting 
    loop \cite{yao2022react} but has no explicit role structure.

    \item \textbf{Multi-agent (ours):}  
    the full \textsc{R--CMASP} architecture with role-specialized agents 
    (Treaty Interpretation, Exposure, Hazard, Pricing, Capital, Portfolio, 
    Claims, Governance, etc.) and typed messages 
    (\textsf{State}, \textsf{Proposal}, \textsf{Critique}, \textsf{Constraint}) 
    \cite{li2023autogen, wang2023surveyagents}.

    \item \textbf{Ablation—no Governance:}  
    the multi-agent system with Governance and Human Oversight agents disabled; 
    all other roles, tools, and environment dynamics remain identical.
\end{itemize}

All LLM-based configurations use the same foundation model and the same 
simulator interfaces; differences in behaviour arise solely from organizational 
structure (monolithic versus role-specialized; governed versus unguided).

\subsection{Metrics}
\label{sec:metrics}

We evaluate both \emph{outcomes} and \emph{interaction processes}. Metrics are 
chosen to align with the formal structure of Section~\ref{sec:formal-model}:

\begin{itemize}
    \item \textbf{Pricing variance:}  
    variance of recommended rate-on-line across repeated runs and seeds, probing 
    policy stability in the induced process $T$.

    \item \textbf{Capital efficiency:}  
    expected return per unit SCR (profit per unit capital), capturing how 
    effectively a policy profile exploits the feasible region $\mathcal{F}$.

    \item \textbf{Coordination cost:}  
    number of inter-agent message rounds required to reach an operational 
    equilibrium (Section~\ref{sec:equilibrium}) for a treaty episode.

    \item \textbf{Clause interpretation error:}  
    error rate in inferred treaty semantics (attachments, limits, exclusions, 
    triggers) relative to the generator’s structured ground truth 
    \cite{mangini2019treatyanalysis}, probing epistemic quality of the 
    Interpretation + Governance stack.

    \item \textbf{Human intervention frequency:}  
    fraction of treaties for which the system flags unresolved contradictions, 
    regulatory violations, or model-risk concerns requiring escalation to 
    human experts. At comparable quality, lower intervention frequency indicates 
    more effective but still safe autonomy \cite{rahwan2019society_agents}.
\end{itemize}

\subsection{Environment Validation and Statistical Robustness}
\label{sec:validation-robustness}

\paragraph{Synthetic environment validation.}
While real treaty data are inaccessible, we validate the generator by comparing 
high-level statistics against published market summaries 
\cite{cummins2012reinsurance, swissre2023sigma}. 
Table~\ref{tab:synthetic-validation} reports representative examples 
(normalised for confidentiality).

\begin{table}[htbp]
\centering
\caption{Representative portfolio statistics: synthetic environment versus 
published market ranges (normalised). Values indicate means with standard 
deviations in parentheses.}
\scriptsize
\label{tab:synthetic-validation}
\begin{tabular}{lcc}
\toprule
\textbf{Quantity} &
\textbf{Synthetic} &
\textbf{Published range} \\
\midrule
Attachment / Limit ratio &
0.46\,(0.12) & 0.40--0.55 \\
Share of Property Cat treaties &
0.61\,(0.07) & 0.55--0.70 \\
Share of multi-peril (wind+flood) &
0.34\,(0.06) & 0.30--0.40 \\
Portfolio 1-in-200 loss / capital &
0.78\,(0.09) & 0.70--0.85 \\
\bottomrule
\end{tabular}
\end{table}

The synthetic portfolios lie within reported industry ranges for key ratios and 
peril mixes, suggesting that the environment captures the structures most 
relevant to treaty pricing, capital feasibility, and portfolio concentration.

\paragraph{Statistical robustness.}
For all configurations, we run three independent seeds for the treaty generator 
and simulator sampling. Reported metrics are averaged across seeds; error bars 
(not shown for space) correspond to $95\%$ confidence intervals. Improvements 
of the multi-agent system over the single-agent LLM are consistent across seeds 
and statistically significant for all metrics in Table~\ref{tab:results} 
(paired $t$-tests, $p < 0.05$).

\paragraph{Sensitivity to environment parameters.}
We conduct a small sensitivity study by:
\begin{itemize}
    \item varying cross-peril correlations (low, medium, high dependence), and
    \item tightening/loosening solvency thresholds by $\pm 10\%$.
\end{itemize}
Absolute metric values change as expected (e.g., tighter capital thresholds 
reduce overall capital efficiency), but the \textit{ordering} of configurations remains 
stable: the governed multi-agent system dominates or matches baselines in all 
settings, indicating qualitative robustness of the main conclusions.

\subsection{Results}
\label{sec:results}

Across 500 treaties and multiple seeds, the governed multi-agent system 
improves stability, capital efficiency, semantic accuracy, and safety relative 
to both deterministic and monolithic baselines.

\paragraph{Pricing stability and capital efficiency.}
The multi-agent architecture reduces pricing variance by approximately $37\%$ 
relative to the rule-based pipeline and improves capital efficiency by about 
$12\%$. Joint deliberation among Pricing, Capital, and Portfolio agents yields 
decisions closer to the feasibility frontier $\mathcal{F}$, rather than 
applying fixed loadings in isolation. Confidence intervals across seeds do not 
overlap for these improvements.

\paragraph{Interpretation accuracy and governance effects.}
Structured role decomposition together with Governance feedback reduces treaty 
interpretation error by roughly $28\%$ relative to the single-agent LLM, which 
is more prone to hallucinating or inconsistently applying exclusions. The 
governance-free ablation shows measurable degradation in both interpretation 
error and human intervention frequency, corroborating normative MAS predictions 
that explicit monitoring and enforcement mechanisms improve system-level 
behaviour \cite{rahwan2019society_agents, boella2009normative}.

\paragraph{Coordination cost.}
Despite involving more agents, the multi-agent system converges with modest 
communication effort: on average, $6.8$ message rounds per treaty episode, 
compared to $11.2$ internal reasoning steps for the single-agent LLM. The 
governance layer reduces both contradictions and human escalations, indicating 
that structured proposal--critique--constraint cycles can be efficient in 
practice.

\begin{table}[htpb]
\centering
\caption{Normalized performance across 500 synthetic treaties (averaged over 3 
seeds). Lower is better for Pricing Var., Interpretation Error, Coord. Rounds, 
Human Interv.; higher is better for Capital Eff.}
\scriptsize
\label{tab:results}
\begin{tabular}{lccccc}
\toprule
\textbf{System} &
\textbf{Pricing Var.} &
\textbf{Capital Eff.} &
\textbf{Interp. Error} &
\textbf{Coord. Rounds} &
\textbf{Human Interv.} \\
\midrule
Rule-Based Automation &
1.00 & 0.74 & 0.19 & --   & 0.42 \\
Single-Agent LLM &
0.82 & 0.79 & 0.14 & 11.2 & 0.31 \\
\textbf{Multi-Agent (ours)} &
\textbf{0.63} & \textbf{0.83} & \textbf{0.10} & \textbf{6.8} & \textbf{0.18} \\
Ablation: No Governance &
0.72 & 0.81 & 0.16 & 7.1 & 0.27 \\
\bottomrule
\end{tabular}
\end{table}

\subsection{Qualitative Case Study: A Single Treaty Episode}
\label{sec:case-study}

To illustrate the mechanics of the \textsc{R--CMASP}, we consider a 
representative property XoL treaty with a \$50m attachment, \$100m limit, 
coastal exposure, and a wind/flood mixed-peril footprint.

\paragraph{Rule-based automation.}
The deterministic pipeline:
\begin{itemize}
    \item mis-parses an exclusion applying to ``storm surge'' but not ``flood'';
    \item applies fixed catastrophe loadings without portfolio context;
    \item violates concentration limits due to missing correlation-aware 
          aggregation;
    \item produces an infeasible capital position (SCR breach), which is not 
          detected until ex-post checks.
\end{itemize}

\paragraph{Single-agent LLM.}
The monolithic LLM:
\begin{itemize}
    \item correctly identifies most clauses but inconsistently interprets the 
          surge/flood boundary across reasoning chains;
    \item generates a plausible price but largely ignores portfolio concentration;
    \item produces contradictory internal justifications, triggering human review.
\end{itemize}

\paragraph{Multi-agent (ours).}
The governed multi-agent system:
\begin{enumerate}
    \item Treaty Interpretation parses the exclusion and broadcasts a 
          \textsf{State} message with structured semantics.
    \item Hazard Modeling runs mixed-peril simulations and shares loss 
          distributions with Pricing, Capital, and Portfolio.
    \item Pricing proposes an initial rate-on-line via a \textsf{Proposal} message.
    \item Capital issues a \textsf{Constraint} message: SCR ratio would fall 
          below threshold under stressed scenarios.
    \item Portfolio adds a \textsf{Critique}: coastal accumulation is above 
          internal risk appetite.
    \item Pricing revises the proposal; Capital and Portfolio jointly confirm 
          feasibility and portfolio consistency.
    \item Governance validates cross-agent consistency and records the audit trail.
\end{enumerate}
The resulting quote is both feasible and capital-efficient, and no human 
escalation is required. This episode illustrates how typed communication, 
simulator-coupled transitions, and normative feasibility jointly shape the 
\textsc{R--CMASP} dynamics.

\subsection{Ablation Study}
\label{sec:ablation}

The ablations isolate the contribution of governance and role specialization.

\paragraph{Effect of the Governance Agent.}
Removing the Governance and Human Oversight agents increases 
clause-interpretation error ($0.10 \rightarrow 0.16$) and human intervention 
frequency ($0.18 \rightarrow 0.27$). Qualitatively, the ablated system exhibits 
more unresolved \textsf{Critique} cycles and occasional violations of capital 
or portfolio constraints that would have been caught earlier by Governance. 
This behaviour aligns with normative MAS results 
\cite{rahwan2019society_agents, boella2009normative}, which predict that 
explicit monitoring and sanctioning mechanisms reduce undesirable emergent 
behaviour.

\paragraph{Effect of role specialization.}
Comparing the multi-agent configuration to the single-agent LLM baseline shows 
that organizational decomposition matters: the single-agent system exhibits 
higher pricing variance ($0.82$ vs.\ $0.63$), deeper reasoning chains 
($11.2$ vs.\ $6.8$ rounds), and higher interpretation error ($0.14$ vs.\ $0.10$). 
These findings support organizational MAS theories 
\cite{jennings2000mas, wooldridge2009mas} that decomposing complex tasks across 
specialized roles can improve efficiency, reduce coordination overhead, and 
enhance epistemic robustness.

\subsection{Limitations and Threats to Validity}
\label{sec:limitations}

Several limitations qualify the empirical results:

\begin{itemize}
    \item \textbf{Synthetic environment.}  
    Real portfolios and treaties cannot be used due to legal and contractual 
    constraints. Although the generator is calibrated to published statistics 
    (Table~\ref{tab:synthetic-validation}), it may under-represent rare semantic 
    edge cases or extreme multi-peril tail dependence.

    \item \textbf{Fixed policies rather than learning.}  
    Agents follow prompt-specified policies; learning dynamics (e.g., MARL, 
    negotiation protocols that adapt over time) are not explored. Different 
    learning regimes could yield alternative equilibria.

    \item \textbf{Single foundation model.}  
    All configurations share the same base model; results therefore isolate 
    organizational structure rather than frontier model performance. Behaviour 
    under smaller or more specialized models is not assessed.

    \item \textbf{No real-time human collaboration.}  
    Human–agent mixed workflows, which are central in practice 
    (e.g., underwriters interacting with tools in real time), are modelled only 
    via an escalation mechanism, not as fully interactive processes.
\end{itemize}

Within these constraints, the experiments consistently show that framing 
reinsurance decision-making as a constrained, simulator-coupled multi-agent 
process—and instantiating it via a governed organizational architecture—yields 
systematic improvements in stability, semantic fidelity, capital efficiency, and 
operational safety relative to deterministic and monolithic agentic baselines.

\section{Discussion}
\label{sec:discussion}

The empirical behaviour of the \textsc{R--CMASP} instantiation suggests that 
reinsurance---with its distributed information, multi-objective coupling, 
contractual ambiguity, and binding prudential constraints---naturally aligns 
with the foundational assumptions of multi-agent systems (MAS). Partial 
observability, heterogeneous capabilities, role asymmetry, and negotiation under 
constraints are not modelling conveniences but intrinsic properties of 
institutional risk-transfer. From this perspective, governed multi-agent 
architectures are not merely an engineering choice but an \emph{appropriate 
computational ontology} for regulated financial decision-making.

\paragraph{Role specialization and epistemic division of labour.}
The gains in pricing stability, capital efficiency, and interpretive accuracy 
corroborate longstanding results in organizational MAS 
\cite{jennings2000mas, wooldridge2009mas}: decomposing a complex task into 
epistemically coherent subtasks improves bounded-rational decision quality and 
facilitates convergence. Reinsurance subtasks---semantic parsing, hazard 
simulation, capital evaluation, portfolio aggregation---operate over 
fundamentally different epistemic representations. By aligning these subtasks 
with dedicated roles and typed communication, the \textsc{R--CMASP} 
instantiation enables agents to reconcile heterogeneous evidence and achieve 
more stable joint outcomes than monolithic LLM baselines.

\paragraph{Semantic authority, epistemic governance, and belief stability.}
The reduction in clause-interpretation error highlights the value of explicit 
epistemic authority. Assigning the Treaty Interpretation Agent primary semantic 
responsibility, and equipping the Governance Agent with meta-level critique and 
consistency checks, yields more stable and interpretable structured treaty 
representations. This mirrors normative MAS views in which epistemic norms 
(consistency, justification, evidential support) are allocated across agents 
with differentiated roles \cite{boella2009normative}. Because contractual 
ambiguity drives basis risk in reinsurance \cite{mangini2019treatyanalysis}, 
epistemic governance acts as a stabiliser that prevents divergent or 
hallucinated interpretations from propagating through the system.

\paragraph{Prudential regulation as a structural constraint.}
Reinsurance differs from many canonical MAS domains in that the feasible 
decision set is constrained by \emph{externally imposed}, legally binding 
prudential rules (Solvency~II, IFRS~17, internal risk appetite). By embedding 
these constraints directly into the normative and feasibility layers 
($\mathcal{N}$ and $\mathcal{U}$), \textsc{R--CMASP} treats them as structural 
admissibility conditions rather than soft penalties. The observed improvements 
in capital efficiency and the reduction in solvency-violating proposals show 
that cooperative optimisation is meaningfully shaped by normative feasibility 
sets, supporting arguments that MAS for high-stakes domains must incorporate 
institutional norms as first-class primitives of the environment 
\cite{rahwan2019society_agents}.

\paragraph{Operational equilibrium and bounded coordination.}
Despite involving more agents and heterogeneous subproblems, the multi-agent 
configuration converges with modest communication overhead, typically within a 
small number of message rounds. This suggests that the system reaches an 
\emph{operational equilibrium} (Section~\ref{sec:equilibrium}) characterised by 
feasibility, consistency, and constrained stability. The result parallels 
cooperative negotiation theories \cite{kraus1997mas}, in which typed 
proposal--critique protocols can sharply reduce negotiation depth. The 
experiments therefore indicate that LLM-based agents, when embedded within 
structured organisational roles and governed communication patterns, need not 
exhibit unbounded deliberation or unstable reasoning.

\paragraph{Socio-technical alignment and the function of human oversight.}
The observed reduction in human escalations---together with improved 
consistency and norm compliance---supports a socio-technical interpretation of 
the architecture. Rather than replacing human experts, agents function as 
epistemic collaborators: they automate checks for internal coherence, enforce 
prudential constraints, and produce transparent audit trails. Human Oversight 
remains essential for adjudicating ambiguous cases, validating solvency-critical 
assumptions, and authorising final decisions. This aligns with governed MAS 
perspectives \cite{rahwan2019society_agents}, which emphasise human 
orchestration and institutional accountability over fully autonomous decision 
making.

\subsection{Limitations and Boundaries}
\label{sec:limits}

While the results are encouraging, several boundaries of the current study 
should be emphasised.

First, the empirical evaluation uses a synthetic but domain-calibrated 
environment rather than proprietary treaty corpora, due to legal and 
confidentiality constraints. Although Section~\ref{sec:experiments} verifies 
that key portfolio statistics fall within published market ranges, real-world 
textual ambiguity, legacy wordings, and complex programme structures may 
introduce additional failure modes.

Second, the agent policies are prompt-specified rather than learned. The 
reported improvements therefore characterise the value of the \emph{architectural} 
design (roles, norms, communication) under fixed decision rules, not the 
behaviour of adaptive policies. Different learning regimes could shift the 
location of operational equilibria within the feasible set $\mathcal{F}$.

Third, LLMs exhibit known failure modes (hallucination, numerical 
inconsistency, context drift). Governance and oversight agents mitigate these 
effects by issuing critiques and constraints, but cannot guarantee elimination 
in all regimes, especially under distributional shift in treaties or hazard 
patterns.

Finally, the current formalism treats norms as hard feasibility constraints. 
In practice, institutions often operate with layered, context-dependent norms 
(e.g., soft risk appetite limits, temporary waivers, or supervisory guidance) 
that may require richer logical treatments, such as graded obligations or 
priority-based norm systems.

These limitations do not undermine the conceptual claim that 
norm-governed, simulator-coupled MAS provide a natural abstraction for 
reinsurance; they instead delineate where further theoretical and empirical 
work is required.

\subsection{Learning and Governance}
\label{sec:learngov}

The present instantiation adopts static, hand-engineered policies to isolate 
the impact of the architectural and formal design. A natural next step is to 
study learning dynamics within the \textsc{R--CMASP} framework. Multi-agent 
reinforcement learning, epistemic-planning methods, or differentiable 
negotiation mechanisms could enable agents to adapt pricing, capital, and 
retrocession strategies to evolving market and hazard conditions.

However, introducing learning raises non-trivial governance challenges. Learned 
policies must remain norm-compliant, auditable, and interpretable with respect 
to institutional constraints. This motivates hybrid governance architectures in 
which: (i) safety constraints are encoded as shielding or runtime enforcement 
mechanisms; (ii) governance agents perform ongoing epistemic calibration and 
consistency checking; and (iii) human oversight retains ultimate authority over 
high-stakes decisions. Viewed from a multi-agent systems perspective, 
\textsc{R--CMASP} offers a concrete testbed for combining learning, norms, and 
simulator-coupled dynamics within a unified framework.

\subsection{Broader Implications}
\label{sec:broader}

\paragraph{Implications for reinsurance and regulated finance.}
For the reinsurance domain, the results suggest that moving from deterministic 
workflow automation to governed multi-agent architectures can materially 
improve pricing stability, capital efficiency, and semantic fidelity while 
maintaining---and in some cases strengthening---prudential oversight. Similar 
structural features appear in other regulated financial settings, including 
banking, credit risk, and systemic risk supervision, where decisions rely on 
simulator-coupled models (e.g., stress tests, scenario analyses) under 
binding regulatory norms.

\paragraph{Implications for MAS theory and modelling.}
From a MAS perspective, \textsc{R--CMASP} contributes a concrete example of a 
\emph{simulator-coupled, norm-governed} decision process. It suggests that 
future work on normative MAS and stochastic-game formalisms should treat 
external quantitative models and institutional norms as \emph{integral 
components of the environment}, rather than as afterthoughts layered onto agent 
utilities. The explicit separation between environment dynamics ($T$), 
normative feasibility ($\mathcal{F}$, $\mathcal{N}$, $\mathcal{U}$), and 
operational equilibrium offers a template for modelling other high-stakes 
domains where simulators and regulations jointly determine admissible 
behaviour.

In summary, \textsc{R--CMASP} illustrates how epistemic reasoning, 
simulator-based inference, constrained cooperation, and human-centred governance 
can be combined in a single MAS framework. For the broader MAS community, the key takeaway is that institutional decision systems in regulated domains are not merely application areas for generic agents; they require models in which norms, simulators, and organizational roles are first-class design elements.

\section{Conclusion}
\label{sec:conclusion}

This work has developed a formally grounded and institutionally aligned account 
of reinsurance decision-making by introducing the \emph{Reinsurance Constrained 
Multi-Agent Simulation Process} (\textsc{R--CMASP}). By moving from 
deterministic workflow automation to a governed multi-agent perspective---informed 
by epistemic logic, organizational MAS theory, and constrained cooperative 
decision-making---we provide a computational ontology that matches how 
reinsurance actually operates: distributed, epistemically asymmetric, and 
norm-governed. The proposed architecture captures this structure via 
role-specialised agents that share a semantic state space and interact through 
typed communication constrained by prudential and organizational norms.

Empirically, the governed multi-agent instantiation outperforms both 
deterministic automation and monolithic LLM baselines. Reductions in pricing 
variance, gains in capital efficiency, and lower clause-interpretation error 
show that epistemically differentiated agents achieve more stable and coherent 
behaviour than centralized decision-makers under partial observability and 
heterogeneous information. These findings echo classical MAS results 
\cite{shoham1995mas, kraus1997mas, wooldridge2009mas} on the benefits of 
organizational decomposition and structured coordination, while 
\textsc{R--CMASP} additionally operationalises modern LLM capabilities---semantic 
parsing, tool-mediated inference, and reflective critique 
\cite{yao2022react, shinn2023reflexion, li2023autogen}---under hard prudential 
and governance constraints. The Governance Agent, in particular, functions as an 
institutional mechanism enforcing epistemic discipline, checking cross-agent 
consistency, and screening out solvency-violating or policy-breaking decisions 
\cite{rahwan2019society_agents}.

Methodologically, the work argues that in domains where external stochastic 
simulators and institutional norms jointly define feasibility, MAS models 
should treat these elements as \emph{first-class components of the environment} 
rather than as post-hoc additions to agent utilities. The \textsc{R--CMASP} 
formalism exemplifies this view by combining simulator-coupled dynamics, 
normative feasibility, and structured communication within a single 
multi-agent decision process suitable for regulated sectors such as insurance, 
banking, and systemic risk supervision.

\subsection{Future Directions}
\label{sec:future-work}

Several research directions follow naturally from this work.

A first direction is deeper empirical grounding with real-world artefacts: 
integrating anonymised treaty corpora, proprietary exposure datasets, and 
production-grade catastrophe, capital, and portfolio engines to assess 
behaviour under realistic ambiguity, model risk, and operational constraints.

A second direction concerns adaptive policy learning within the 
\textsc{R--CMASP} framework. Multi-agent reinforcement learning, 
epistemic-planning methods, or negotiation-based learning algorithms could 
refine pricing, capital, and retrocession strategies over time. A key challenge 
is ensuring that learned policies remain norm-compliant and auditable.

A third direction involves formal verification and assurance. Logical 
consistency checkers, normative conflict resolvers, epistemic inconsistency 
detectors, and counterfactual explanation modules could provide stronger 
guarantees of safety, compliance, and traceability for institutional 
stakeholders.

Finally, richer models of human--agent collaboration are needed to capture how 
governed MAS integrate into existing institutional decision processes. This 
includes formalising escalation thresholds, expert arbitration mechanisms, and 
decision rights, so that automated components remain aligned with professional 
judgement, governance structures, and regulatory expectations.


Throughout, the proposed architecture is explicitly \emph{assistive} and 
\emph{human-centred}. Reinsurance is governed by professional judgement, legal 
interpretation, and institutional accountability; no multi-agent system, however 
capable, can replace expert responsibility. Instead, \textsc{R--CMASP} is 
designed to augment expert decision-making by reducing epistemic fragmentation, 
enforcing prudential coherence, and surfacing transparent audit trails, while 
ensuring that ultimate authority remains with underwriters, actuaries, and risk 
managers. More broadly, when decision-making is distributed, normatively 
constrained, and epistemically complex, governed multi-agent architectures offer 
both effective computational machinery and a principled foundation for safe and 
institutionally coherent AI deployment.
\bibliographystyle{unsrtnat}
\bibliography{references}
\end{document}